\documentclass[letterpaper, 10pt, conference]{ieeeconf} 

\IEEEoverridecommandlockouts                              

\usepackage[noadjust]{cite}
\usepackage{graphics} 
\usepackage{epsfig} 
\usepackage{times} 
\usepackage{amssymb, amsmath}
\usepackage{dsfont}
\usepackage[font=small,labelfont=bf]{caption}
\usepackage[labelformat=simple]{subcaption}

\usepackage{breqn}
\usepackage{array}
\usepackage{float}
\usepackage{siunitx}
\usepackage{mathtools}
\usepackage{float}
\usepackage{siunitx}

\usepackage{amsthm}
\usepackage{xspace}
\usepackage[linesnumbered,ruled,vlined]{algorithm2e}
\usepackage{array}
\usepackage{algorithmic}
\usepackage{makecell}
\usepackage{booktabs}

\usepackage{threeparttable}
\usepackage{xcolor}

\newtheorem{theorem}{Theorem}

\newtheorem{definition}{Definition}
\newtheorem{remark}{Remark}
\newtheorem{assumption}{Assumption}
\newtheorem{proposition}{Proposition}

\DeclareMathOperator*{\argmin}{arg\,min}

\newcommand{\R}{\mathbb{R}}

\title{\LARGE \bf
Confidence-Aware Safe and Stable Control of Control-Affine Systems 
}

\author{Shiqing Wei, Prashanth Krishnamurthy, and Farshad Khorrami
\thanks{The authors are with Control/Robotics Research Laboratory, Department of Electrical and Computer Engineering, NYU Tandon School of Engineering, 5 Metrotech Center, Brooklyn, NY 11201, USA. 
{\tt\small \{shiqing.wei, prashanth.krishnamurthy, khorrami\}@nyu.edu}}%
\thanks{This work was supported in part by ARO grants W911NF-21-1-0155 and W911NF-22-1-0028 and by the New York University Abu Dhabi (NYUAD) Center for Artificial Intelligence and Robotics (CAIR), funded by Tamkeen under the NYUAD Research Institute Award CG010.}
}

\begin{document}
\maketitle
\thispagestyle{empty}
\pagestyle{empty}

\begin{abstract}
Designing control inputs that satisfy safety requirements is crucial in safety-critical nonlinear control, and this task becomes particularly challenging when full-state measurements are unavailable. In this work, we address the problem of synthesizing safe and stable control for control-affine systems via output feedback (using an observer) while reducing the estimation error of the observer. To achieve this, we adapt control Lyapunov function (CLF) and control barrier function (CBF) techniques to the output feedback setting. Building upon the existing CLF-CBF-QP (Quadratic Program) and CBF-QP frameworks, we formulate two confidence-aware optimization problems and establish the Lipschitz continuity of the obtained solutions. To validate our approach, we conduct simulation studies on two illustrative examples. The simulation studies indicate both improvements in the observer's estimation accuracy and the fulfillment of safety and control requirements.
\end{abstract}

\section{Introduction}
For safety-critical control problems, it is of paramount importance to design controllers that not only satisfy the system performance requirements, but also ensure the system's safety during operation. Recently, control barrier functions (CBFs) have emerged as a popular approach to generate safe control inputs for a wide range of control tasks \cite{ames2019control, wei2024diffocclusion, dai2023safe, dai2023data, gonccalves2024safe, gonccalves2023using}. In general, the CBF is designed based on a designated safe set, and the controls are generated to make this safe set forward invariant, ensuring the safety of the system. However, this task becomes particularly challenging in scenarios where full-state measurements are unavailable, and one instead has to rely on partial information about the system states. In this work, we address the problem of formulating controls that can both guarantee safety via output feedback (using an observer) and improve confidence about the states.

To achieve these objectives, we employ an EKF (Extended Kalman Filter) based nonlinear observer \cite{reif1998ekf} and adapt control Lyapunov functions (CLFs) and CBFs to the setting of output feedback. We extend the existing quadratic programs (QPs), such as CLF-CBF-QP and CBF-QP, and design control inputs that can speed up the observer's convergence and meet the necessary safety and control requirements.

Synthesizing safe control based on output measurements is an ongoing research area \cite{cosner2022self, agrawal2022safe, wang2022observer, clark2021control, wei2020towards}. In \cite{cosner2022self} and \cite{wei2020towards}, the authors develop motion-planning systems to generate safe trajectories based on image sensor/laser measurements. Other works (e.g., \cite{agrawal2022safe, wang2022observer}) have focused on employing observers with a quantified estimation error and solving QPs for a safe controller. We inherit the definition of safe output-feedback controllers and observer-robust CBFs from \cite{agrawal2022safe}. In \cite{clark2021control}, CBFs are formulated for stochastic systems with incomplete state information. However, these approaches primarily follow a ``top-down'' paradigm in the sense that no interaction with the observer is made in the control design process. Typically, the observer's estimation error is influenced by the control inputs, and in this study, we demonstrate that safety can be consistently achieved by selecting controls that enhance the estimation accuracy of the observer.

Our work is also related to active sensing or observability optimization, where a ``bottom-up'' approach is adopted to reduce uncertainty or enhance system observability. In \cite{hinson2013path}, optimal paths are identified to improve the system observability for under-sensed vehicles in a planar uniform flow field. The works \cite{salaris2019online} and \cite{napolitano2022information} propose perception-aware trajectory generation methods aimed at maximizing the information collected by output measurements for autonomous robots. In \cite{coleman2021observability}, the authors address the problem of observability-aware target tracking for mobile robots using a nonlinear model predictive control framework. In contrast to the aforementioned studies, our work incorporates both stability and safety requirements into the search for control inputs that enhance system observability.

\textit{Our Contributions:} (1) We present an optimization-based control approach that addresses the design of safe and stabilizing controls for control-affine nonlinear systems using output feedback, specifically focusing on enhancing state confidence. (2) We extend the existing CLF-CBF-QP and CBF-QP frameworks and formulate two optimization problems incorporating confidence-aware considerations. (3) We prove the feasibility of these optimization problems and demonstrate the Lipschitz continuity of the obtained solutions. (4) We demonstrate the effectiveness of our approach through simulation studies on two illustrative examples: a second-order nonlinear system stabilization problem and a unicycle tracking problem. The results indicate notable improvements in the observer's estimation accuracy, alongside the successful fulfillment of safety and control requirements.

\section{EKF-Based Nonlinear Observer}\label{sec:ekf}
Consider the following nonlinear system with dynamics
\begin{equation}\label{eq:nonlinear_sys}
    \dot{x} = p(x,u), \quad z = q(x)
\end{equation}
where $x \in \R^{n_x}$ is the state, $u \in \R^{n_u}$ is the control input, and $z \in \R^{n_z}$ is the output. The functions $p: \R^{n_x} \times \R^{n_u} \rightarrow \R^{n_x}$ and $q: \R^{n_x} \rightarrow \R^{n_z}$ are assumed to be locally Lipschitz and $\mathcal{C}^2$ functions. An EKF-based nonlinear observer for system \eqref{eq:nonlinear_sys} is proposed in \cite{reif1998ekf}
\begin{equation} \label{eq:observer}
    \dot{\hat{x}} = p(\hat{x},u) + K(t)[z-q(\hat{x})]
\end{equation}
where $\hat{x} \in \R^{n_x}$ is the estimated state and the time-varying observer gain $K(t)$ is an $n_x \times n_z$ matrix. Denote by $A(t)$ and $C(t)$ the following partial differentials
\begin{equation}\label{eq:partials}
    A(t) = \frac{\partial p}{\partial x}(\hat{x}, u) \quad \text{and} \quad C(t) = \frac{\partial q}{\partial x}(\hat{x}).
\end{equation}
For $\kappa \geq 0$ and symmetric positive definite matrices $Q \in \R^{n_x \times n_x}$ and $R \in \R^{n_z \times n_z}$, the observer gain is defined as 
\begin{equation}
    K(t) = P(t) C^\top(t) R^{-1}
\end{equation}
where $P(t)$ is the solution to the Riccati equation 
\begin{equation}\label{eq:riccati_p}
    \dot{P} = \kappa P + AP + PA^\top - PC^\top R^{-1} CP + Q.
\end{equation}
We drop the time dependence of $A(t)$, $C(t)$, and $P(t)$ for simplicity when it does not cause confusion.

\begin{assumption}\label{ass:ass_p}
    There exist two constants $\underline{p}, \bar{p} > 0$ such that
\begin{equation}
    \underline{p} I \leq  P(t) \leq \bar{p} I, \quad \forall t \geq 0.
\end{equation}
\end{assumption}

The above assumption is made in many works on EKF-based observers (e.g., \cite{reif1998ekf, bonnabel2014contraction}). As pointed out in \cite{bonnabel2014contraction}, this assumption can be practically checked in the following way: the user keeps track of the bounds $\underline{p}(t)$ and $\bar{p}(t)$ such that $\underline{p}(t) I < P(s) < \bar{p}(t) I$ for $s \leq t$ and verifies that the bound on the estimation error associated with $\underline{p}(t)$ and $\bar{p}(t)$ holds (at least) up to time $t$. Analogous to the EKF in the probabilistic setting, we call $P(t)$ the \textit{uncertainty} of the estimated states and $S(t) = P^{-1}(t)$ the \textit{confidence} of the observer. Noting that $\dot{S}(t) = -P^{-1} (t) \dot{P} (t) P^{-1} (t)$, we obtain the dynamics of the confidence $S(t)$ by rearranging \eqref{eq:riccati_p}: 
\begin{equation}\label{eq:riccati_s}
    \dot{S} = -\kappa S - A^\top S - SA + C^\top R^{-1} C - SQS.
\end{equation}

\section{Observer-Based Safe and Stable Control}
\subsection{Local Exponential Stability of the Observer}
Consider the plant
\begin{equation}\label{eq:control_affine_sys}
    \dot{x} = f(x) +g(x)u, \quad z = q(x)
\end{equation}
where we assume that $f: \mathcal{X} \rightarrow \R^{n_x}$, $g : \mathcal{X} \rightarrow \R^{n_x \times n_u}$, and $q: \mathcal{X} \rightarrow \R^{n_z}$ are of class $\mathcal{C}^2$, and denote $\mathcal{I} = [0, t_{\text{max}})$ as the maximal interval of existence. The state $x$, control $u$, and output $z$ are of dimensions $n_x$, $n_u$, and $n_z$, respectively. To model the physical constraints of the real world, we assume that the state evolves in a compact set $x \in \mathcal{X} \subset \R^{n_x}$, while the control input is unconstrained, i.e., $u \in \R^{n_u}$. We further assume that the origin is an equilibrium of \eqref{eq:control_affine_sys} for $u=0$, $q(0) =0$, and $0 \in \mathcal{X}$.

As \eqref{eq:control_affine_sys} is a special case of \eqref{eq:nonlinear_sys}, the observer \eqref{eq:observer} is equally applicable to system \eqref{eq:control_affine_sys}. Recall the dynamics of the estimated state $\hat{x}(t)$ using the confidence $S(t)$ are 
\begin{equation}\label{eq:observer_recall}
    \dot{\hat{x}} = f(\hat{x}) + g(\hat{x})u + S^{-1}(t) C^\top(t) R^{-1}[z-q(\hat{x})],
\end{equation}
and the partial differentials in \eqref{eq:partials} become
\begin{equation}\label{eq:partials_affine}
    A(t) = \frac{\partial f}{\partial x}(\hat{x}) + \frac{\partial g}{\partial x}(\hat{x})u \quad \text{and} \quad C(t) = \frac{\partial q}{\partial x}(\hat{x})
\end{equation}
where $\frac{\partial f}{\partial x}$ and $\frac{\partial q}{\partial x}$ are Jacobian matrices, and $\frac{\partial g}{\partial x}$ is a tensor\footnote{$\frac{\partial f}{\partial x}$ and $\frac{\partial q}{\partial x}$ are $n_x \times n_x$ and $n_z \times n_x$ matrices respectively. $\frac{\partial g}{\partial x}$ is a three-dimensional tensor of size $n_x \times n_x \times n_u$. The product $\frac{\partial g}{\partial x} u$ is a $n_x \times n_x$ matrix.}. As $R$ is a fixed positive definite matrix chosen by the user, it can be bounded by $\underline{r}I \leq R \leq \bar{r}I$ with $\underline{r}, \bar{r} >0$. We further assume that the estimated state $\hat{x} \in \hat{\mathcal{X}}$ and $\hat{\mathcal{X}}$ is a compact subset of $\R^{n_x}$. Under an output-feedback controller $\pi: \mathcal{I} \times \hat{\mathcal{X}} \times \R^{n_z} \rightarrow \R^{n_u}$, the closed-loop system \eqref{eq:control_affine_sys} along with the observer is
\begin{align}\label{eq:control_affine_sys_cl}
    \dot{x} &= f(x) +g(x)\pi(t,\hat{x},z), \quad z = q(x), \notag \\
    \dot{\hat{x}} &= f(\hat{x}) + g(\hat{x})\pi(t,\hat{x},z) + S^{-1} C^\top R^{-1}[z-q(\hat{x})].
\end{align}

Denote the estimation error of the observer by 
\begin{equation}\label{eq:est_error}
    \zeta(t) = x(t) - \hat{x}(t).
\end{equation}
Although the control input is unconstrained, the exponential convergence of the observer requires the applied control signal to be bounded, since $A(t)$ in \eqref{eq:partials_affine} depends affinely on $u$.

\begin{assumption}\label{ass:bounded_u}
    The control signal applied to \eqref{eq:control_affine_sys} is bounded, i.e., there exists a constant $\bar{u} > 0$ such that $\lVert u(t) \rVert \leq \bar{u}$ for all $t \in \mathcal{I}$.
\end{assumption}

Assumption \ref{ass:bounded_u} is not restrictive: we show in Remark \ref{rmk:bounded_pi} that the controllers obtained in Section \ref{sec:conf_opt} satisfy it. It is proved in \cite{reif1998ekf} that the estimation error $\zeta(t)$ is locally exponentially convergent to zero.

\begin{proposition}[\cite{reif1998ekf}]\label{prop:exp_stable}
Under Assumptions \ref{ass:ass_p} and \ref{ass:bounded_u}, the boundedness of $\mathcal{X}$ and $\hat{\mathcal{X}}$, and the $\mathcal{C}^2$-smoothness of $f$, $g$, and $q$, the EKF-based observer in \eqref{eq:observer_recall} is a local exponential observer. More specifically, there exist positive real numbers $\epsilon, \eta> 0$ and $\theta > \kappa/2$ such that
\begin{equation}\label{eq:exp_error}
    \lVert \zeta(t) \rVert \leq \eta \lVert \zeta(0) \rVert e^{-\theta t}
\end{equation}
for $t \geq 0$ with $\zeta(0) \in B_\epsilon$ where $B_\epsilon = \{v \in \R^{n_x} \mid \lVert v \rVert < \epsilon \}$. 
\end{proposition} 

\subsection{Control Lyapunov Functions}
Control Lyapunov functions (CLFs) are commonly used to prove a closed-loop system's stability. In the context of output feedback, we introduce the following definition.

\begin{definition}\label{def:clf}
For system \eqref{eq:control_affine_sys} and the observer \eqref{eq:observer_recall} with known estimation error bound \eqref{eq:exp_error}, a class $\mathcal{C}^2$ positive definite function $V: \mathcal{X} \cup \hat{\mathcal{X}} \rightarrow \R_+$ is an observer-based exponentially stabilizing CLF, if there exist a constant $\gamma > 0$ and two class $\mathcal{K}$ functions $\alpha_1, \alpha_2$ such that $\forall x \in \mathcal{X} \cup \hat{\mathcal{X}}$
\begin{align}
    &\ \alpha_1(\lVert x \rVert) \leq V(x) \leq \alpha_2(\lVert x \rVert), \\
    &\inf_{u \in \R^{n_u}} \left(L_f V(x) + L_g V(x) u + \gamma V(x)\right) \leq 0
\end{align}
where $L_f V(\cdot) = \nabla V^\top (\cdot) f(\cdot)$ and $L_g V(\cdot) = \nabla V^\top (\cdot) g(\cdot)$ are the Lie derivatives of $V$ w.r.t. $f$ and $g$, respectively, and $\nabla V: \mathcal{X} \cup \hat{\mathcal{X}} \rightarrow \R^{n_x}$ is the gradient of $V$.
\end{definition}
Compared with the definition of the CLF in \cite{sontag1989universal} and \cite{wei2023neural}, Definition \ref{def:clf} requires the additional exponential stability of the Lyapunov function as in \cite{ames2019control}. As we are working with the estimated state $\hat{x}$ given by the observer, we further extend the domain of $V$ to $\mathcal{X} \cup \hat{\mathcal{X}}$. The $\mathcal{C}^2$-smoothness is required by later analysis in Section \ref{sec:conf_opt}. Next, consider the set 
\begin{equation}\label{eq:control_clf}
    K_{\text{clf}}(\hat{x}) \! = \! \left\{ u\in \R^{n_u} \! \mid \! L_f V(\hat{x}) \! + \! L_g V(\hat{x}) u \! + \! \gamma V(\hat{x}) \! \leq \! 0 \right\}.
\end{equation}
The following result shows that given a sufficiently accurate initial state of the observer, a Lipschitz continuous output-feedback controller $\pi(\hat{x}) \in K_{\text{clf}}(\hat{x})$ renders system \eqref{eq:control_affine_sys} asymptotically stable. 

\begin{theorem}
Assume that the conditions of Proposition \ref{prop:exp_stable} hold and the initial estimation error satisfies $x(0)-\hat{x}(0) \in B_\epsilon$ as in Proposition \ref{prop:exp_stable}. If such a CLF $V$ exists as in Definition \ref{def:clf}, then any Lipschitz continuous controller $\pi(\hat{x}) \in K_{\text{clf}}(\hat{x})$ asymptotically stabilizes system \eqref{eq:control_affine_sys} with the observer \eqref{eq:observer_recall} of known estimation error bound \eqref{eq:exp_error}.
\end{theorem}
\begin{proof}
Since the conditions of Proposition \ref{prop:exp_stable} hold and $x(0)-\hat{x}(0) \in B_\epsilon$, we have $\lVert x(t) - \hat{x}(t) \rVert \leq M(t)$ with $M(t) = \eta \lVert \zeta(0) \rVert e^{-\theta t}$. Noting that $V$ is continuously differentiable and $\mathcal{X} \cup \hat{\mathcal{X}}$ is compact, we denote $K_V$ as the Lipschitz constant of $V$. Then, for $t \geq 0$, we have
$
    V(x(t)) \leq V(\hat{x}(t)) + K_V M(t) \coloneqq W(t).
$
By \eqref{eq:observer_recall}, we have
\begin{align}\label{eq:W_dot}
    \dot{W}(t) &= L_f V(\hat{x}) + L_g V(\hat{x}) u + K_V \dot{M}(t) \nonumber\\
    + &\nabla V(\hat{x})^\top S^{-1}(t) C^\top(t) R^{-1}[q(x)-q(\hat{x})].
\end{align}
Denote by $K_q$ the Lipschitz constant of $q$ and by the definition of $C(t)$ in \eqref{eq:partials_affine}, we have $\lVert C(t) \rVert < K_q$. Note that $\lVert \nabla V(\hat{x}) \rVert \leq K_V$, $\lVert S^{-1}(t) \rVert \leq \bar{p}$ (by Assumption \ref{ass:ass_p}), and $\lVert R^{-1} \rVert < \underline{r}^{-1}$, then it follows
\begin{equation}\label{eq:bound_thm1}
    \left \lVert \nabla V(\hat{x})^\top S^{-1}(t) C^\top(t) R^{-1}[q(x)-q(\hat{x})] \right \rVert \leq b M(t)
\end{equation}
where $b = \underline{r}^{-1} \bar{p} K_V K_q^2$. Considering the output-feedback controller $\pi(\hat{x}) \in K_{\text{clf}}(\hat{x})$ and the bound in \eqref{eq:bound_thm1}, one has 
\begin{align*}
    \dot{W}(t) &\leq -\gamma V(\hat{x}(t)) + bM(t) + K_V \dot{M}(t) \\
    &= -\gamma W(t) + (b+\gamma K_V)M(t) + K_V \dot{M}(t).
\end{align*}
Construct the following ODE
\begin{equation}\label{eq:ode_thm1}
    \dot{y} = -\gamma y + (b+\gamma K_V)M(t) + K_V \dot{M}(t),\ y(0) = W(0).
\end{equation}
Then, we have $W(t) \leq y(t)$ for $t \geq 0$ by Comparison Lemma \cite[Lemma B.2]{khalil2015nonlinear}. Since $M(t) = \eta \lVert \zeta(0) \rVert e^{-\theta t}$ and $\dot{M}(t) = -\theta \eta \lVert \zeta(0) \rVert e^{-\theta t}$, we can solve for $y(t)$.

If $\gamma = \theta$, $y(t) = (W(0)+b \eta \lVert \zeta(0) \rVert t)e^{-\gamma t}$. Since $V(x(t)) \leq W(t) \leq y(t)$ and $V(x) \geq \alpha_1(\lVert x \rVert)$, we have 
\begin{equation}\label{eq:thm_x_bound1}
    \lVert x(t) \rVert \leq \alpha_1^{-1} \left( (W(0)+b \eta \lVert \zeta(0) \rVert t)e^{-\gamma t} \right).
\end{equation}

If $\gamma \neq \theta$, $y(t) = W(0)e^{-\gamma t} + \frac{c}{\gamma - \theta} (e^{-\theta t} - e^{-\gamma t})$ with $c = \eta \lVert \zeta(0) \rVert(b+\gamma K_V - \theta K_V)$. Similarly, we have 
\begin{equation}\label{eq:thm_x_bound2}
    \lVert x(t) \rVert \leq \alpha_1^{-1} \left( W(0)e^{-\gamma t} + \frac{c}{\gamma - \theta} (e^{-\theta t} - e^{-\gamma t}) \right).
\end{equation}

As $\alpha_1^{-1}$ is also a class $\mathcal{K}$ function and thus continuous, it follows that in both cases, $\lVert x(t) \rVert \rightarrow 0$ as $t \rightarrow \infty$ by \eqref{eq:thm_x_bound1} and \eqref{eq:thm_x_bound2}, i.e., the system \eqref{eq:control_affine_sys} is asymptotically stable.
\end{proof}

\subsection{Control Barrier Functions}
We say that the system \eqref{eq:control_affine_sys} is \textit{safe} if the true state $x(t)$ stays within the safe set $\mathcal{S}$ characterized by 
\begin{equation}
    \mathcal{S} = \{x \in \mathcal{X} \mid h(x) \geq 0 \}
\end{equation}
where $h: \mathcal{X} \cup \hat{\mathcal{X}} \rightarrow \R$ is a $\mathcal{C}^2$ function and $\operatorname{Int}(\mathcal{S}) \neq \varnothing$. When the knowledge of the full state $x$ is available, as commonly assumed in the existing literature \cite{ames2019control}, one tries to find a state-feedback controller that renders $\mathcal{S}$ forward invariant, i.e., $x(0) \in \mathcal{S} \Rightarrow x(t) \in \mathcal{S}, \forall t \in \mathcal{I}$. However, in the setting of output-feedback control, we need to ensure the safety of the true state $x(t)$ using only $\hat{x}(t)$.

\begin{definition}[\cite{agrawal2022safe}]
An output-feedback controller $\pi: \mathcal{I} \times \hat{\mathcal{X}} \times \R^{n_z} \rightarrow \R^{n_u}$ renders system \eqref{eq:control_affine_sys} safe w.r.t. the set $\mathcal{S}$ if for the closed-loop system \eqref{eq:control_affine_sys_cl}
\begin{equation}
    x(0) \in \mathcal{X}_0 \text{ and } \hat{x}(0) \in \hat{\mathcal{X}}_0 \Rightarrow x(t) \in \mathcal{S}, \forall t \in \mathcal{I}
\end{equation} 
where $\mathcal{X}_0 \subset \mathcal{S}$ and $\hat{\mathcal{X}}_0 \subset \hat{\mathcal{X}}$ are sets of initial conditions for $x(t)$ and $\hat{x}(t)$, respectively.
\end{definition}

By Proposition \ref{prop:exp_stable}, we have $\lVert \zeta(t) \rVert \leq M(t)$ with $M(t) = \eta \lVert \zeta(0) \rVert e^{-\theta t}$ if  $\zeta(0) \in B_\epsilon$. Then, it follows that
\begin{align*}
    \dot{h}(\hat{x}) &= L_f h(\hat{x}) + L_g h(\hat{x}) u \\
    &\quad + \nabla h(\hat{x})^\top S^{-1}(t) C^\top(t) R^{-1}[q(x)-q(\hat{x})]\\
    &\geq L_f h(\hat{x}) + L_g h(\hat{x}) u -  \underline{r}^{-1} \bar{p} K_h K_q^2 M(0)
\end{align*}
where $K_h$ is the Lipschitz constant of $h$ on $\mathcal{X} \cup \hat{\mathcal{X}}$, $\lVert C(t) \rVert < K_q$, $\lVert S^{-1}(t) \rVert \leq \bar{p}$, and $\lVert R^{-1} \rVert < \underline{r}^{-1}$. Next, we introduce the definition of a CBF in the context of output feedback. 

\begin{definition}[Adapted from \cite{agrawal2022safe}]\label{def:cbf}
A class $\mathcal{C}^2$ function $h: \mathcal{X} \cup \hat{\mathcal{X}} \rightarrow \R$ is an observer-based CBF for system \eqref{eq:control_affine_sys} with the observer \eqref{eq:observer_recall} of known estimation error bound \eqref{eq:exp_error} if there exists a constant $0< \alpha \leq \theta$ such that for all $x \in \mathcal{S}$
\begin{equation*}
    \sup_{u \in \R^{n_u}} \left (L_f h(x) + L_g h(x) u - \underline{r}^{-1} \bar{p} K_h K_q^2 M(0) \right) \geq -\alpha h(x).
\end{equation*} 
\end{definition}

For a given CBF $h$, consider the set
\begin{align}\label{eq:control_cbf}
    K_{\text{cbf}}&(t,\hat{x},z) = \{ u\in \R^{n_u} \mid L_f h(\hat{x}) + L_g h(\hat{x}) u + \alpha h(\hat{x}) \nonumber \\
    & +\nabla h(\hat{x})^\top S^{-1}(t) C^\top(t) R^{-1}[z-q(\hat{x})] \geq 0
    \}.
\end{align}
The next result shows that the controller $\pi(t,\hat{x},z) \in K_{\text{cbf}}(t,\hat{x},z)$ renders system \eqref{eq:control_affine_sys} safe if the initial condition $x(0)$ is relatively far from the boundary of the safe set $\mathcal{S}$ and $\hat{x}(0)$ is close to $x(0)$. We assume that system \eqref{eq:control_affine_sys} is of relative degree one, i.e., $L_g h(x) \neq 0$ for $x \in \mathcal{X} \cup \hat{\mathcal{X}}$. As the control input is unconstrained, the supremum in Definition \ref{def:cbf} is attained trivially wherever $L_g h(x) \neq 0$; the condition is restrictive only at points where $L_g h(x) = 0$, which the relative degree one assumption excludes. The same observation applies to the CLF condition in Definition \ref{def:clf}, which remains restrictive at points where $L_g V(x) = 0$, such as the origin.

\begin{theorem}\label{thm:cbf}
Assume the conditions of Proposition \ref{prop:exp_stable} hold and a CBF $h$ exists as in Definition \ref{def:cbf}. For system \eqref{eq:control_affine_sys} with the observer \eqref{eq:observer_recall} of known estimation error bound \eqref{eq:exp_error}, if
\begin{align}
    x(0) \in \mathcal{X}_0 &= \{ x \in \mathcal{S} \mid h(x) \geq 2 K_h M(0)\}, \label{eq:x_initial} \\
    \hat{x}(0) \in \hat{\mathcal{X}}_0 &= \{ \hat{x} \in \hat{\mathcal{X}} \mid x(0) - \hat{x}(0) \in B_\epsilon \}, \label{eq:xhat_initial}
\end{align}
then any Lipschitz continuous controller $\pi(t,\hat{x},z) \in K_{\text{cbf}}(t,\hat{x},z)$ renders system \eqref{eq:control_affine_sys} safe w.r.t. the safe set $\mathcal{S}$.
\end{theorem}
\begin{proof}
Since $\hat{x}(0) \in \hat{\mathcal{X}}_0$, it follows that $h(x(t)) \geq h(\hat{x}(t)) - K_h M(t) \coloneqq H(t)$. Given that $x(0) \in \mathcal{X}_0$, we have
\begin{equation*}
    H(0) = h(\hat{x}(0)) - K_h M(0) \geq  h(x(0)) - 2 K_h M(0) \geq 0.
\end{equation*}
Since $\pi(t,\hat{x},z) \in K_{\text{cbf}}(t,\hat{x},z)$, we have
\begin{align*}
    \dot{H} &= L_f h(\hat{x}) + L_g h(\hat{x}) u - K_h \dot{M}(t)\\
    &\quad \quad+ \nabla h(\hat{x})^\top S^{-1}(t) C^\top(t) R^{-1}[z-q(\hat{x})]\\
    &\geq -\alpha h(\hat{x}) - K_h \dot{M}(t) = -\alpha H - \alpha K_h M(t) - K_h \dot{M}(t).
\end{align*}
Noting that $M(t) = \eta \lVert \zeta(0) \rVert e^{-\theta t}$ and $0< \alpha \leq \theta$, we have $- \alpha K_h M(t) - K_h \dot{M}(t) \geq 0$ and thus $\dot{H} \geq -\alpha H$. Given that $H(0) \geq 0$ and $h(x(t)) \geq H(t)$, we have $h(x(t)) \geq H(t) \geq 0$, i.e., system \eqref{eq:control_affine_sys} is safe w.r.t. the safe set $\mathcal{S}$.
\end{proof}

In this work, by $\alpha \leq \theta$, we require that the observer has a faster convergence than the CBF, and this aligns with the commonly accepted principle that the observer should always converge faster than the controller \cite{agrawal2022safe}. The benefit of this control design in \eqref{eq:control_cbf} is that it does not require explicitly calculating $M(t)$. If one explicitly knows $M(t)$, another way to design a safe controller is presented in \cite{agrawal2022safe}.

\section{Confidence Optimization}\label{sec:conf_opt}
As introduced in Section \ref{sec:ekf}, $P(t)$ is analogous to the covariance of the state in the probabilistic setting. If we can optimize some metric of $P(t)$ by selecting proper control inputs $u$, we can speed up the convergence of the observer and thus improve the performance of the feedback controller. 

Recall that the \textit{confidence} matrix is defined by $S(t) = P^{-1}(t)$. We choose $\lambda_{\text{min}}(S(t))$ as the optimization metric, where $\lambda_{\text{min}}(\cdot)$ denotes the minimal eigenvalue of a square matrix. If $\lambda_{\text{min}}(S(t))$ can be increased, then we increase the convergence rate for the slowest mode of the observer. Let $\Delta t$ be the time difference between two consecutive control inputs. At time $t$, we generate control inputs to maximize $\lambda_{\text{min}}(S(t+\Delta t))$, i.e., to optimize $\lambda_{\text{min}}(S(t))$ at one step into the future. In addition, as $S(t)$ satisfies the Riccati equation \eqref{eq:riccati_s}, $S(t+\Delta t)$ can be approximated by 
\begin{align}\label{eq:S_next}
    S(t+\Delta t) &\approx S(t) + \Delta t [-\kappa S(t) - A(t)^\top S(t) - S(t)A(t) \nonumber\\
    & + C(t)^\top R^{-1} C(t) - S(t)QS(t)]
\end{align}
using the first-order approximation. 

\begin{assumption}\label{ass:eig_diff}
    All eigenvalues of $P(t)$ (or equivalently, $S(t)$) are distinct.
\end{assumption}

Assumption \ref{ass:eig_diff} is a practical condition to guarantee well-posedness of the problem. In general, for a real symmetric matrix $X$, $\lambda_{\text{min}}(X)$ is a concave function of $X$ and Lipschitz continuous in $X$ (see \cite[Ch. 2]{wilkinson1971algebraic}). If Assumption \ref{ass:eig_diff} holds, we obtain twice differentiability of $\lambda_{\text{min}}(X)$ with respect to $X$ \cite{overton1995second}. It is worth noting that the set of positive definite matrices with distinct eigenvalues is dense in the set of all positive definite matrices, as we can always make small perturbations to the entries of a positive definite matrix with repeated eigenvalues to obtain a positive definite matrix with distinct eigenvalues.

\begin{remark}
    Apart from the minimum eigenvalue, other metrics, such as the condition number, trace, and determinant, can be equally employed as optimization objectives. We opt for the minimum eigenvalue as it specifically addresses the slowest mode of the estimation error.
\end{remark}

\subsection{Combining CLFs and CBFs via Convex Optimization}
We next formulate the confidence-aware safe and stable control as a constrained optimization problem. The main benefit of optimization-based control is that it allows us to optimize a certain performance objective subject to both stability and safety requirements. More specifically, given a CLF $V$ (Definition \ref{def:clf}) and a CBF $h$ (Definition \ref{def:cbf}) associated with a safe set $\mathcal{S}$, they can be incorporated into finding a single controller $\pi$ that can optimize the confidence of the observer through the optimization problem \eqref{eq:P1a}:
\begin{equation}\label{eq:P1a}
\begin{aligned}
\pi(t, \hat{x}&, z) = \argmin_{u \in \R^{n_u}, \delta \in \R} \  u^\top u -c_1 \lambda_{\text{min}}(S(t+\Delta t)) + c_2 \delta^2 \\
\textrm{s.t.} \quad & L_f V(\hat{x}) + L_g V(\hat{x}) u + \gamma V(\hat{x}) \leq \delta, \\
  & L_f h(\hat{x}) + L_g h(\hat{x}) u + \alpha h(\hat{x}) \\
  & \qquad +\nabla h(\hat{x})^\top S^{-1}(t) C^\top(t) R^{-1}[z-q(\hat{x})] \geq 0  \\
\end{aligned}
\tag{P1}
\end{equation}
where $c_1 \geq 0$, $c_2 >0$, and $\delta \in \R$ is a relaxation variable. The CLF is taken as a soft constraint as in \cite{ames2019control}, while the CBF is taken as a hard constraint. This is a convex problem because $-\lambda_{\text{min}}(S(t+\Delta t))$ is convex w.r.t. $S(t+\Delta t)$ and $S(t+\Delta t)$ has affine dependence on $u$. The following result proves the Lipschitz continuity and safety of the controller $\pi$ given by the optimization problem \eqref{eq:P1a}.

\begin{theorem}\label{thm:p1a}
    Consider system \eqref{eq:control_affine_sys} and observer \eqref{eq:observer_recall} of a known error bound \eqref{eq:exp_error}. Suppose that Assumptions \ref{ass:ass_p} and \ref{ass:eig_diff} and conditions of Proposition \ref{prop:exp_stable} hold, $V$ is a CLF, and $h$ is a CBF associated with the safe set $\mathcal{S}$. If the initial conditions $x(0)$ and $\hat{x}(0)$ satisfy \eqref{eq:x_initial} and \eqref{eq:xhat_initial}, then the controller $\pi: \mathcal{I} \times \hat{\mathcal{X}} \times \R^{n_z} \rightarrow \R^{n_u}$ given by \eqref{eq:P1a} renders system \eqref{eq:control_affine_sys} safe and is piecewise continuous w.r.t. $t$ and Lipschitz continuous w.r.t. $\hat{x}$ and $z$.
\end{theorem}
\begin{proof}
    We first prove existence and uniqueness of the solution to \eqref{eq:P1a}. Let $a_1(\hat{x}) = L_g V(\hat{x})$, $b_1(\hat{x}) = -L_f V(\hat{x}) - \gamma V(\hat{x})$, $a_2(\hat{x}) = -L_g h(\hat{x})$, and $b_2(t, \hat{x}, z) = L_f h(\hat{x}) + \alpha h(\hat{x}) + \nabla h(\hat{x})^\top S^{-1} C^\top R^{-1}[z-q(\hat{x})]$. We omit their dependencies in the following and use $a_1, a_2, b_1$, and $b_2$ for brevity. The constraints in \eqref{eq:P1a} can be written as $T [u^\top, \delta]^\top \leq [b_1, b_2]^\top$ with $T = [a_1, -1; a_2, 0]$ and we see that the rows of $T$ are linearly independent. As there are $n_u + 1$ (with $n_u \geq 1$) decision variables and two linearly independent constraints, the problem is feasible. Since the objective function is strongly convex, there exists one unique minimizer to \eqref{eq:P1a}.

    Then, we prove the Lipschitz continuity of $\pi$. As the objective function of \eqref{eq:P1a} is twice differentiable and strongly convex, its Hessian is positive definite. As the constraints are linearly independent, the regularity conditions of \cite[Thm. D.1]{hager1979lipschitz} are met. Therefore, $\pi$ is Lipschitz continuous w.r.t. the data $a_1, a_2, b_1, b_2, A, C$ and $S(t)$. As the state-dependent data are all Lipschitz continuous w.r.t. $\hat{x}$\footnote{This results from the boundedness of $\mathcal{X}$ and $\hat{\mathcal{X}}$ and the $\mathcal{C}^2$-smoothness of $f,g,q,h,$ and $V$.}, $b_2$ and $S(t)$ are piecewise continuous in $t$, and $b_2$ is Lipschitz continuous in $z$, we see that $\pi$ is piecewise continuous w.r.t. $t$ and Lipschitz continuous w.r.t. $\hat{x}$ and $z$. Finally, $\pi$ renders system \eqref{eq:control_affine_sys} safe w.r.t. $\mathcal{S}$ because \eqref{eq:control_cbf} holds as a hard constraint in \eqref{eq:P1a} and the conditions of Theorem \ref{thm:cbf} are met.
\end{proof}

\subsection{Tracking a Nominal Controller}
In some cases, we may already have a nominal output-feedback controller $\pi_n$ and would like to optimize the confidence of the observer while guaranteeing safety. In this case, we can consider the following problem \eqref{eq:P2}:
\begin{equation}\label{eq:P2}
\begin{aligned}
\pi(t, \hat{x}&, z) = \argmin_{u \in \R^{n_u}} \ \lVert u-\pi_n (\hat{x}) \rVert^2 -c_1 \lambda_{\text{min}}(S(t+\Delta t))  \\
\textrm{s.t.} \quad 
  & L_f h(\hat{x}) + L_g h(\hat{x}) u + \alpha h(\hat{x})  \\
  & \qquad +\nabla h(\hat{x})^\top S^{-1}(t) C^\top(t) R^{-1}[z-q(\hat{x})] \geq 0  \\ 
\end{aligned}
\tag{P2}
\end{equation}
where the CBF is incorporated as a hard constraint, and the objective function is a weighted sum of the tracking error and the cost on the smallest eigenvalue of $S(t+\Delta t)$. If $c_1 = 0$, we recover the CBF-QP as in \cite{ames2019control}.

\begin{theorem}\label{thm:p2}
    Consider system \eqref{eq:control_affine_sys} and observer \eqref{eq:observer_recall} of a known error bound \eqref{eq:exp_error}. Suppose that Assumptions \ref{ass:ass_p} and \ref{ass:eig_diff} and conditions of Proposition \ref{prop:exp_stable} hold, and $h$ is a CBF associated with the safe set $\mathcal{S}$. If the initial conditions $x(0)$ and $\hat{x}(0)$ satisfy \eqref{eq:x_initial} and \eqref{eq:xhat_initial}, and the nominal controller $\pi_n$ is Lipschitz continuous w.r.t. its argument, then the controller $\pi: \mathcal{I} \times \hat{\mathcal{X}} \times \R^{n_z} \rightarrow \R^{n_u}$ given by \eqref{eq:P2} renders system \eqref{eq:control_affine_sys} safe and is piecewise continuous w.r.t. $t$ and Lipschitz continuous w.r.t. $\hat{x}$ and $z$.
\end{theorem}
\begin{proof}
There are $n_u \geq 1$ decision variables and one constraint, so the problem \eqref{eq:P2} is feasible. Since the objective function is strongly convex, there exists a unique minimizer to the problem. Denote $a_1(\hat{x}) = -L_g h(\hat{x})$ and $b_1(t, \hat{x}, z) = L_f h(\hat{x}) + \alpha h(\hat{x}) + \nabla h(\hat{x})^\top S^{-1}(t) C^\top(t) R^{-1}[z-q(\hat{x})]$. Similarly, $\pi$ is Lipschitz continuous w.r.t. the data $a_1, \allowbreak b_1, \allowbreak A, \allowbreak C, \allowbreak S(t)$, and $\pi_n$. Note that the nominal controller $\pi_n$ is Lipschitz continuous w.r.t. $\hat{x}$. In addition, as the previous arguments in the proof of Theorem \ref{thm:p1a} still hold, $\pi$ is piecewise continuous w.r.t. $t$ and Lipschitz continuous w.r.t. $\hat{x}$ and $z$. The rest of the proof is identical to that of Theorem \ref{thm:p1a}.
\end{proof}

\begin{remark}\label{rmk:bounded_pi}
    The controllers given by \eqref{eq:P1a} and \eqref{eq:P2} satisfy Assumption \ref{ass:bounded_u}, so that the boundedness of the applied control is a consequence of the design rather than an additional restriction. Indeed, $S(t+\Delta t)$ in \eqref{eq:S_next} is affine in $u$ and $\lambda_{\text{min}}(\cdot)$ is Lipschitz, whence $\lambda_{\text{min}}(S(t+\Delta t)) \leq a_0 + a_1 \lVert u \rVert$ for some $a_0, a_1 >0$, and the objective functions of \eqref{eq:P1a} and \eqref{eq:P2} are coercive in $u$. Moreover, the true output satisfies $z = q(x) \in q(\mathcal{X})$, which is compact, and $\ell \coloneqq \min_{\hat{x} \in \hat{\mathcal{X}}} \lVert L_g h(\hat{x}) \rVert > 0$ by the relative degree one assumption. Denoting by $B>0$ a bound on $b_2$ (resp. $b_1$) over $\hat{\mathcal{X}} \times q(\mathcal{X})$, the point $u^\star = L_g h(\hat{x})^\top B/ \lVert L_g h(\hat{x}) \rVert^2$ is feasible and satisfies $\lVert u^\star \rVert \leq B/\ell$. Comparing the optimal value with the value attained at $u^\star$ yields a constant $r>0$ such that $\lVert \pi(t,\hat{x},z) \rVert \leq r$ for all $(t,\hat{x},z) \in \mathcal{I} \times \hat{\mathcal{X}} \times q(\mathcal{X})$. As $r$ depends only on $f, g, h, q, V$ over the compact sets $\mathcal{X}$ and $\hat{\mathcal{X}}$, on $\bar{p}$, and on the design constants, and not on $\bar{u}$ itself, Assumption \ref{ass:bounded_u} holds with $\bar{u} = r$.
\end{remark}

\section{Simulation Studies}
\subsection{A Second-Order Nonlinear System}\label{sec:sec_order_sys}
Consider the following second-order nonlinear system 
\begin{equation}\label{eq:eg1_dynamics}
    \dot{x}_1 = -x_1/4 -x_2, \quad
    \dot{x}_2 = x_1^3 - x_2/2 + (x_2^2+1) u
\end{equation}
with output $z = x_1$. A CLF for this system is $V(x) = x_1^4/4 + x_2^2/2$, and we can verify that $\dot{V}(x) = -V(x)$ for $u=0$. The CBF chosen for this system is $h(x) = -x_1/2 + x_2 + 0.5$, and we would like to make sure that the system remains in the closed half-plane where $h(x) \geq 0$.

In Fig.~\ref{fig:p1a}, the system is controlled using the solution to the optimization problem  \eqref{eq:P1a}. The case without confidence optimization ($c_1=0$) is analogous to the setting in \cite{agrawal2022safe}. From Figs.~\ref{fig:p1a_traj_wo_opt}, \ref{fig:p1a_traj_w_opt} and \ref{fig:p1a_control}, we can see that the solution with confidence optimization ($c_1=1000$) gives different control inputs that lead to a different system trajectory, but the system remains safe and stable in both cases. In Figs.~\ref{fig:p1a_eigenvalues} and \ref{fig:p1a_ekf_error}, the larger eigenvalue $\lambda_{\max}$ of $P(t)$ is reduced and the estimation error $x_2 -\hat{x}_2$ decreases faster. In fact, for this example, both of the eigenvalues of $P(t)$ are reduced.

\begin{figure*}[t]
    \centering
    \begin{subfigure}[b]{0.19\textwidth}
        \centering
        \includegraphics[width=\textwidth]{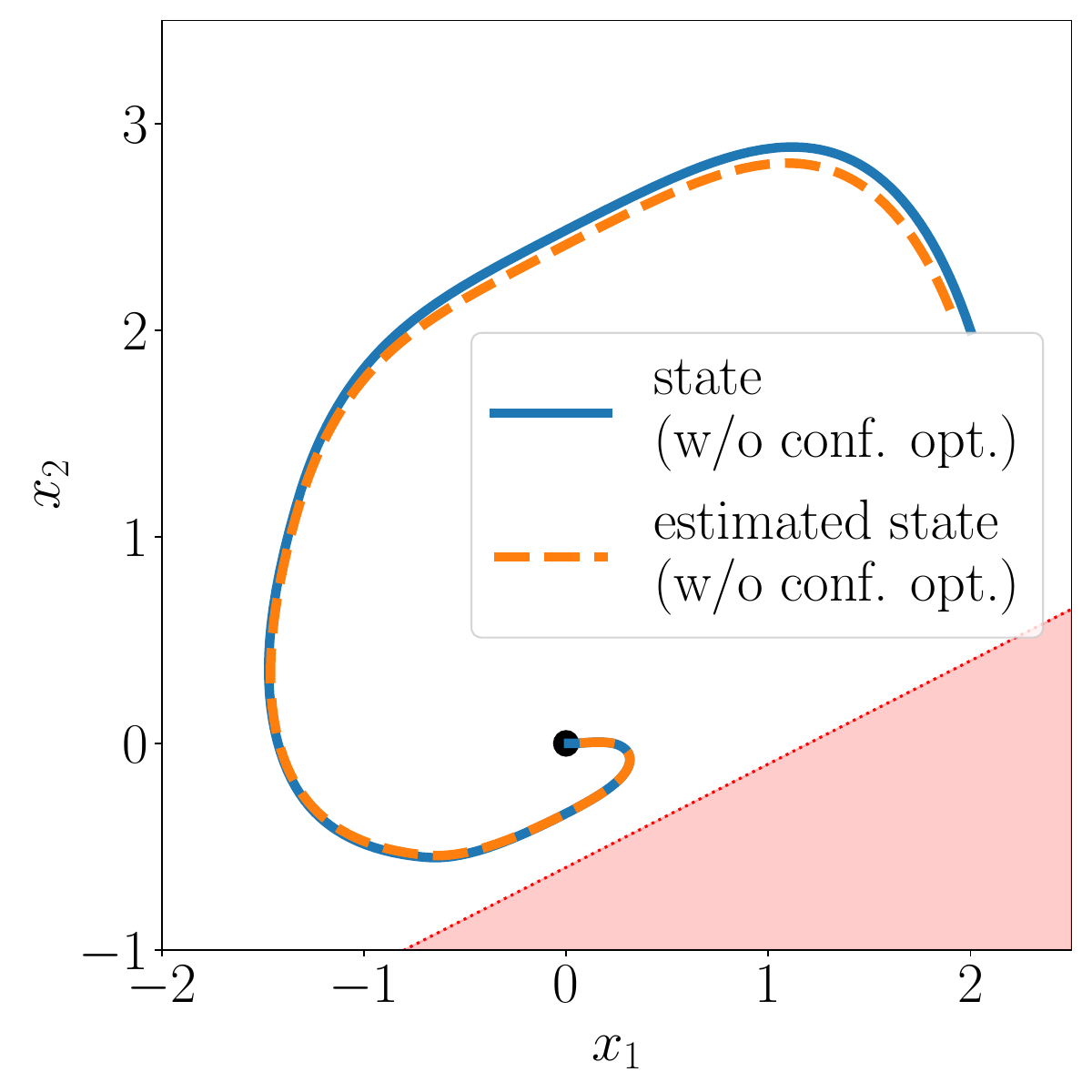}
        \caption{Trajectory ($c_1 = 0$).}
        \label{fig:p1a_traj_wo_opt}
    \end{subfigure}
    \begin{subfigure}[b]{0.19\textwidth}
        \centering
        \includegraphics[width=\textwidth]{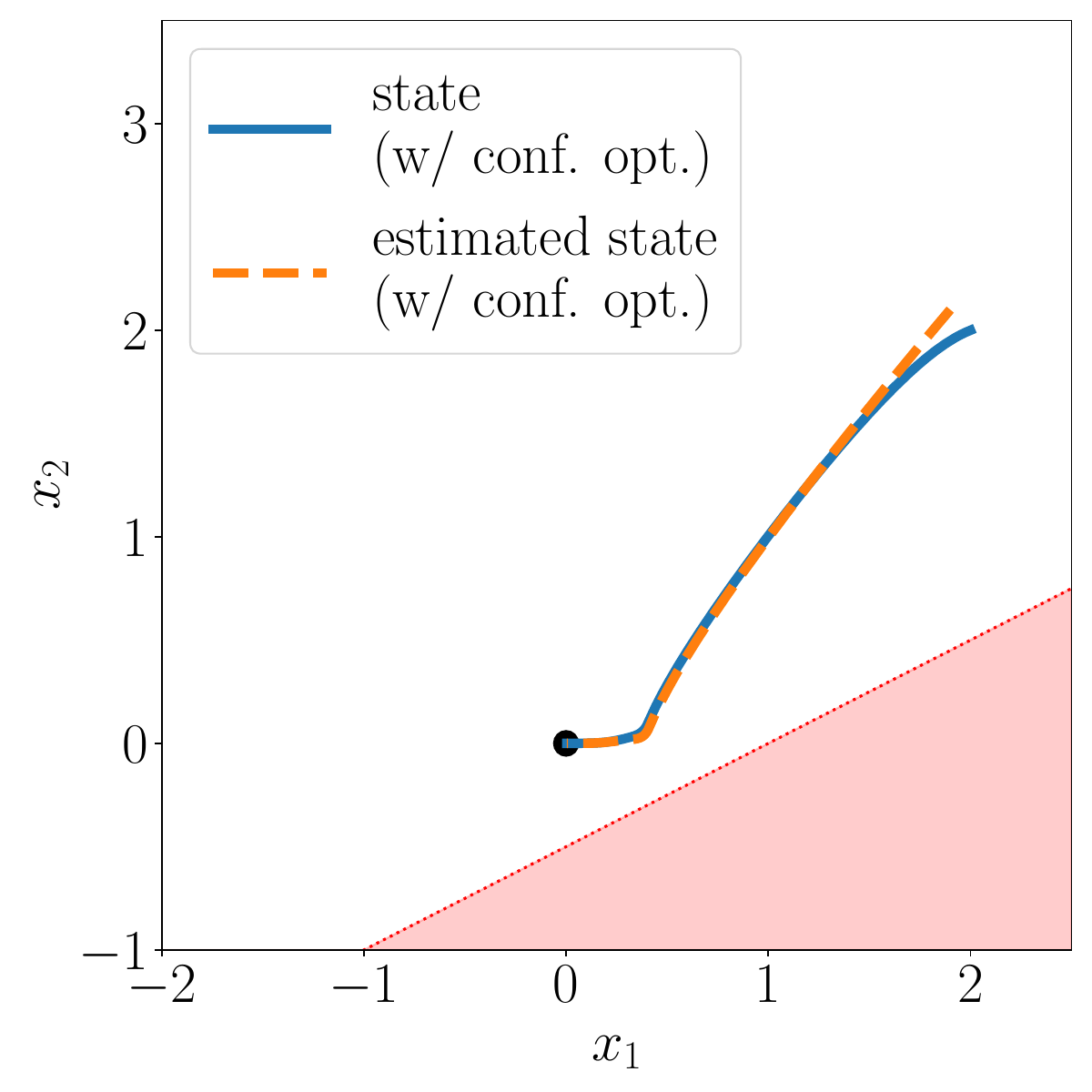}
        \caption{Trajectory ($c_1 = 1000$).}
        \label{fig:p1a_traj_w_opt}
    \end{subfigure}
    \begin{subfigure}[b]{0.19\textwidth}
        \centering
        \includegraphics[width=\textwidth]{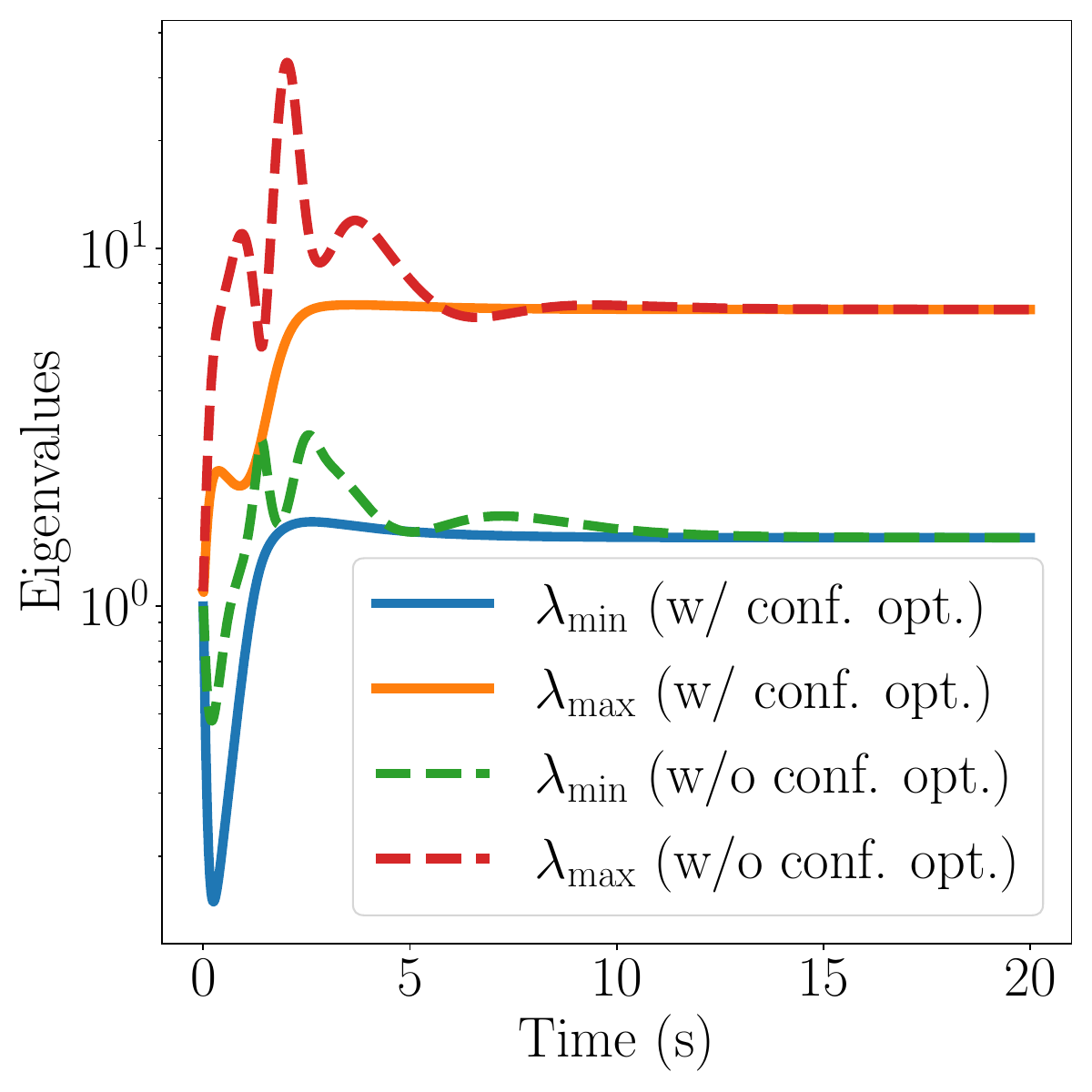}
        \caption{Eigenvalues of $P(t)$.}
        \label{fig:p1a_eigenvalues}
    \end{subfigure}
    \begin{subfigure}[b]{0.19\textwidth}
        \centering
        \includegraphics[width=\textwidth]{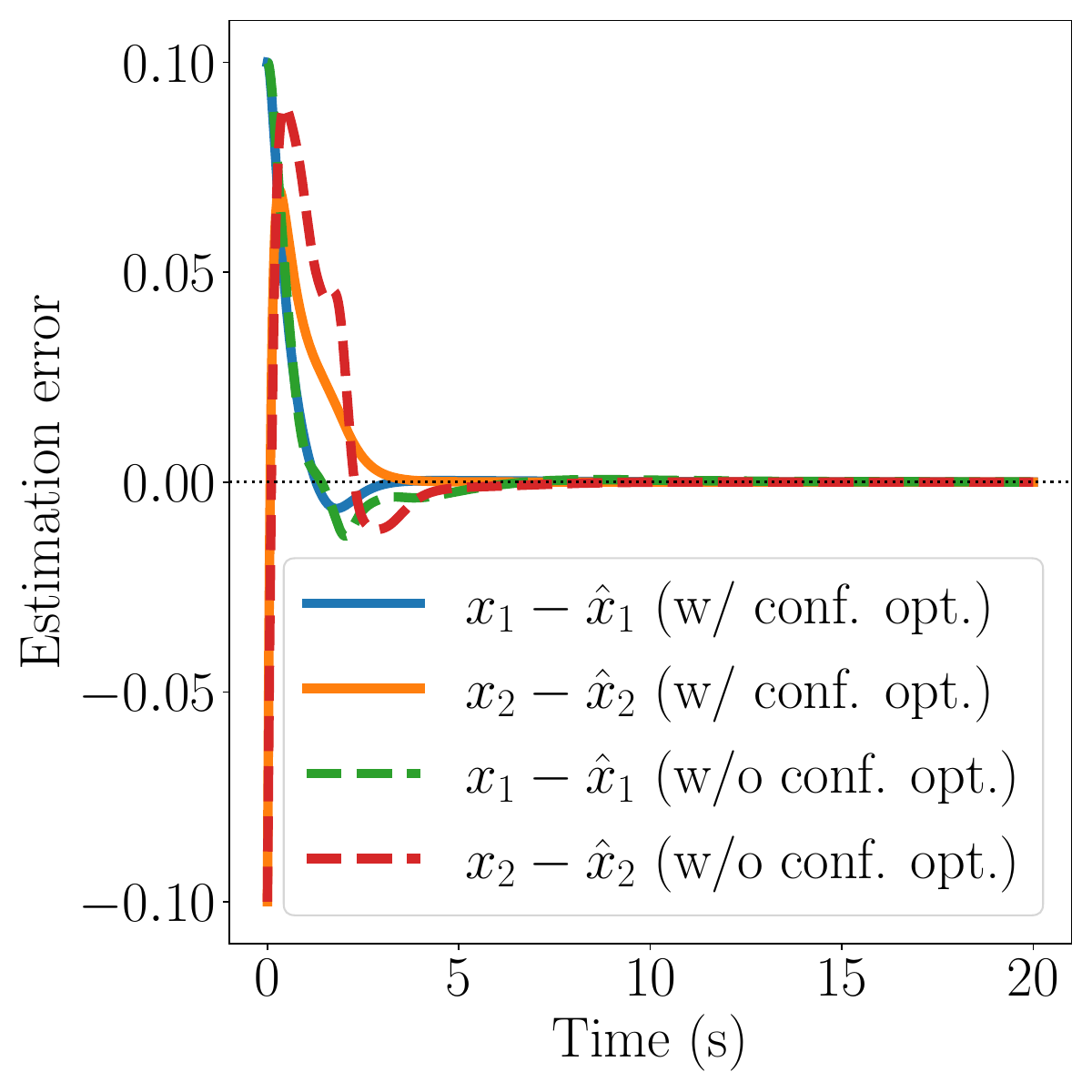}
        \caption{State estimation error.}
        \label{fig:p1a_ekf_error}
    \end{subfigure}
    \begin{subfigure}[b]{0.19\textwidth}
        \centering
        \includegraphics[width=\textwidth]{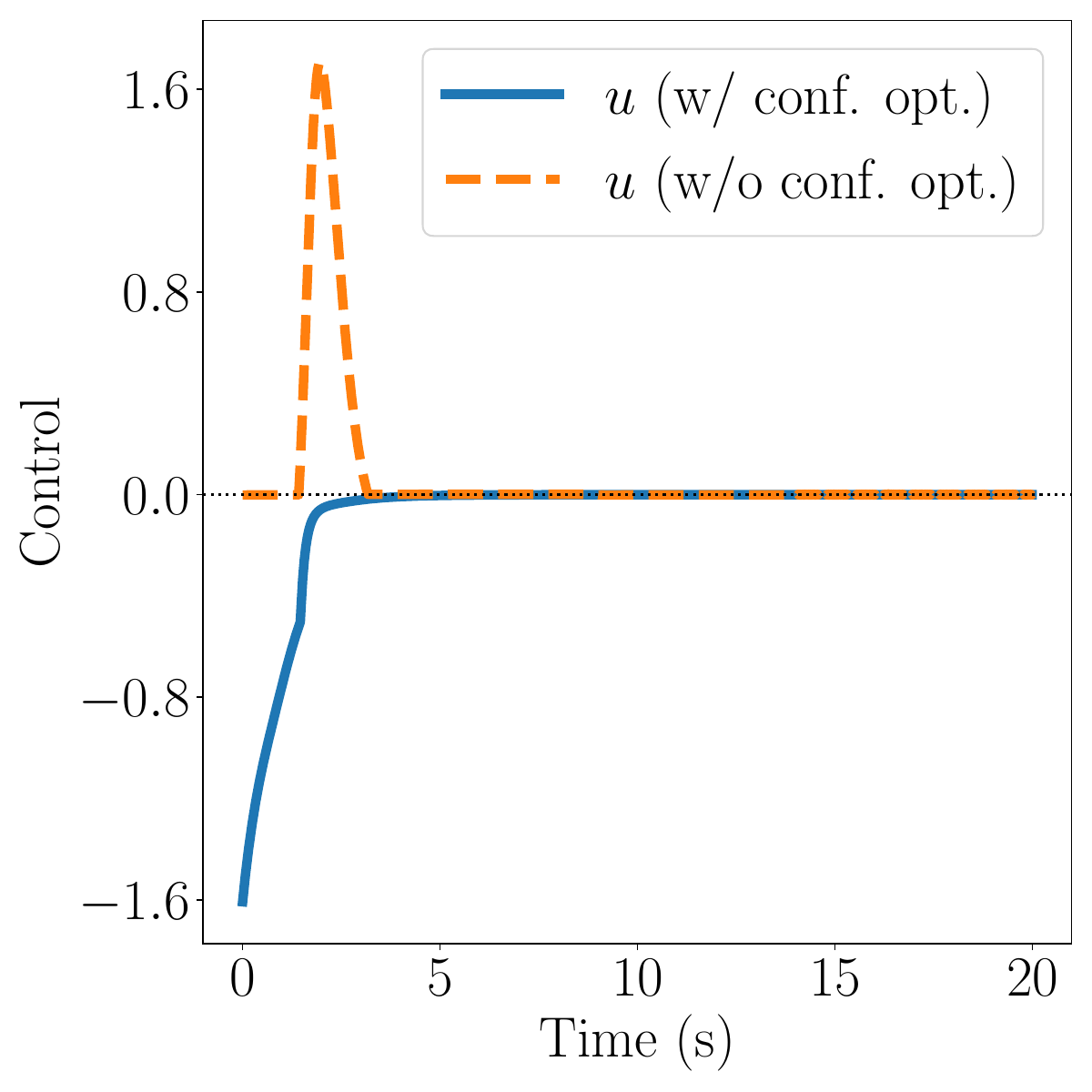}
        \caption{Controls.}
        \label{fig:p1a_control}
    \end{subfigure}
    
    \caption{A second-order system stabilization problem. \subref{fig:p1a_traj_wo_opt} and \subref{fig:p1a_traj_w_opt}: Trajectories without and with confidence optimization. \subref{fig:p1a_eigenvalues}-\subref{fig:p1a_control}: Comparison of the eigenvalues (of $P(t)$), the state estimation error, and the control inputs, respectively. 
    }
    \label{fig:p1a}
\end{figure*}

\begin{figure*}[t]
    \centering
    \begin{subfigure}[b]{0.19\textwidth}
        \centering
        \includegraphics[width=\textwidth]{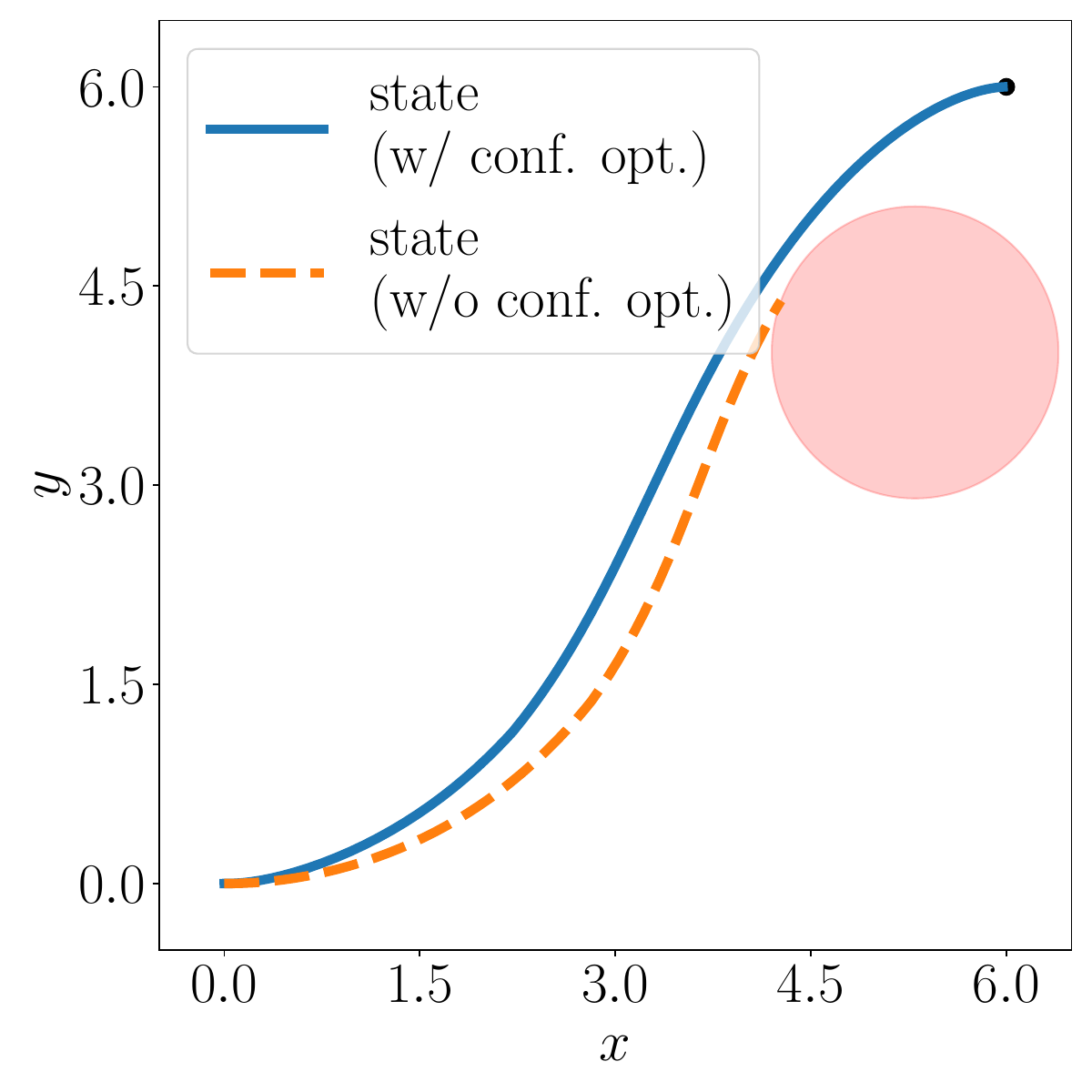}
        \caption{Trajectories.}
        \label{fig:p2_trajs}
    \end{subfigure}
    \begin{subfigure}[b]{0.19\textwidth}
        \centering
        \includegraphics[width=\textwidth]{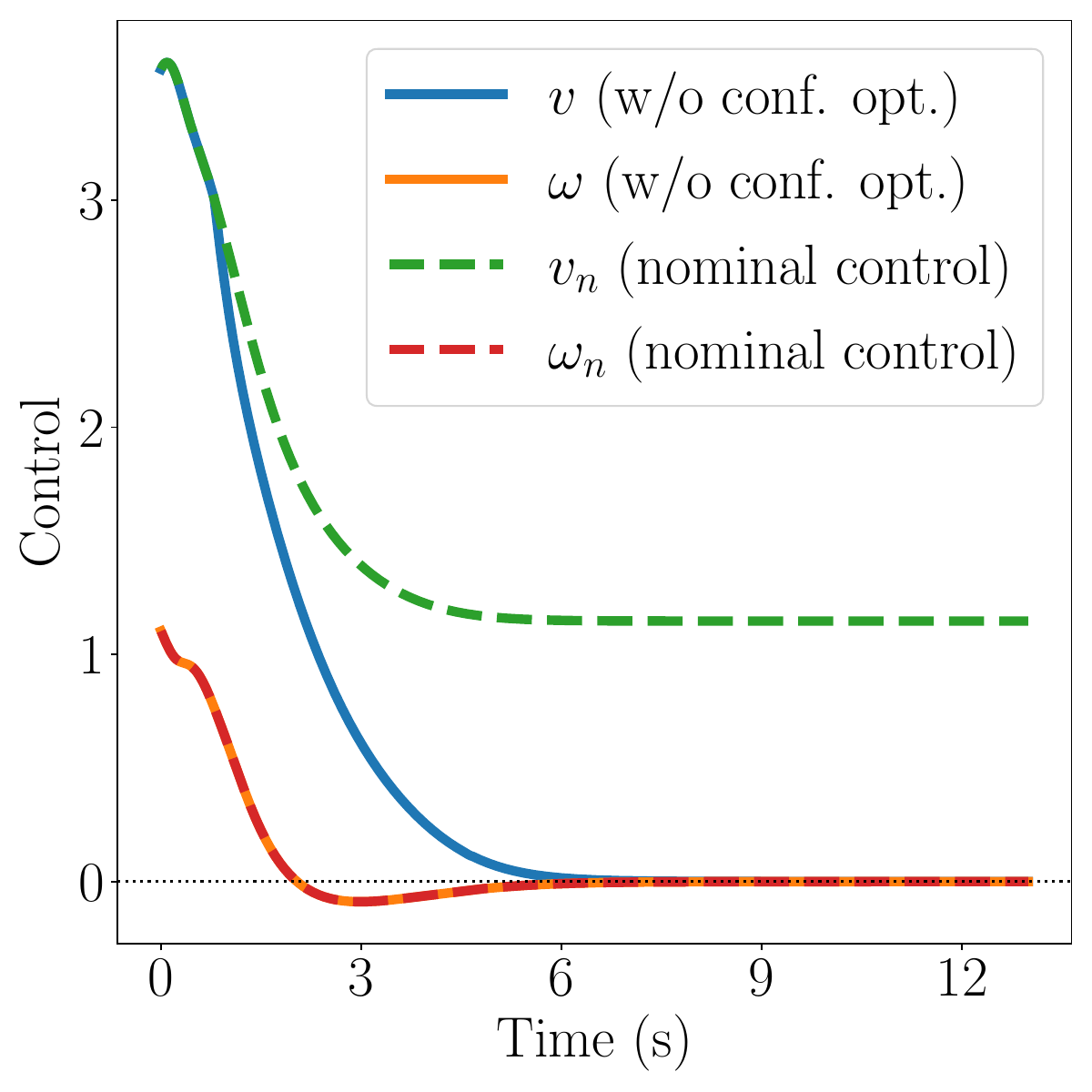}
        \caption{Controls ($c_1 = 0$).}
        \label{fig:p2_control_wo_opt}
    \end{subfigure}
    \begin{subfigure}[b]{0.19\textwidth}
        \centering
        \includegraphics[width=\textwidth]{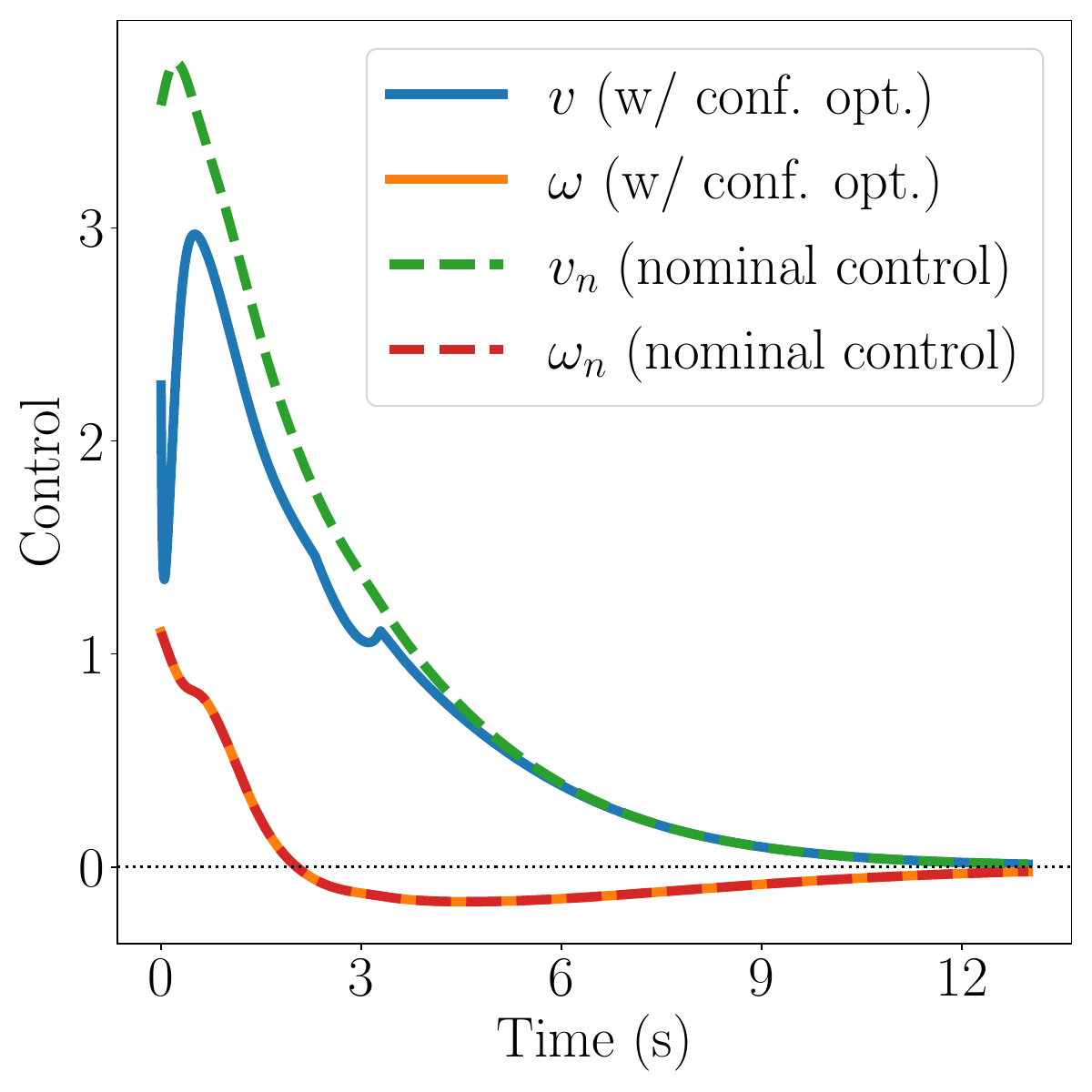}
        \caption{Controls ($c_1 = 1000$).}
        \label{fig:p2_control_w_opt}
    \end{subfigure}
    \begin{subfigure}[b]{0.19\textwidth}
        \centering
        \includegraphics[width=\textwidth]{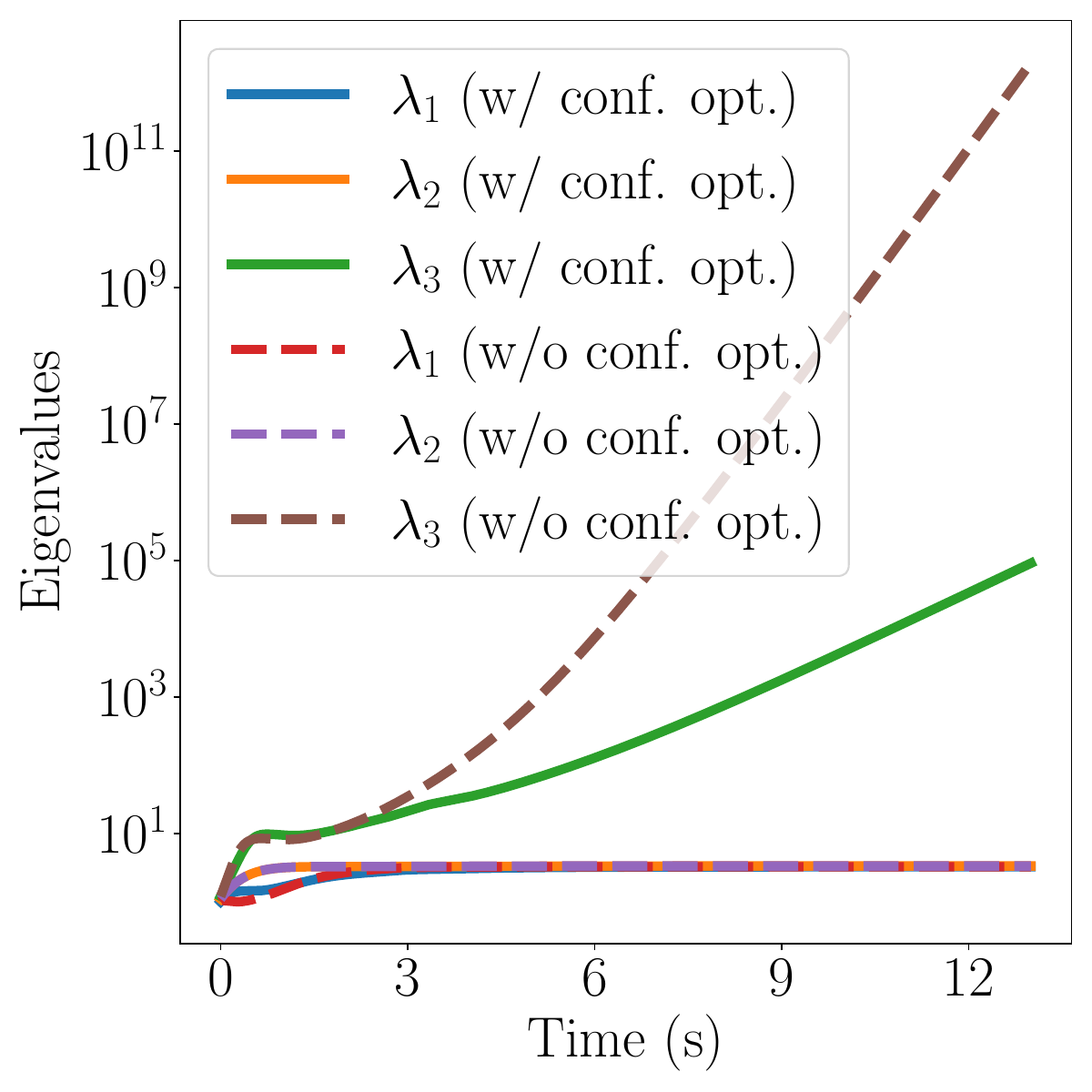}
        \caption{Eigenvalues of $P(t)$.}
        \label{fig:p2_eigenvalues}
    \end{subfigure}
    \begin{subfigure}[b]{0.19\textwidth}
        \centering
        \includegraphics[width=\textwidth]{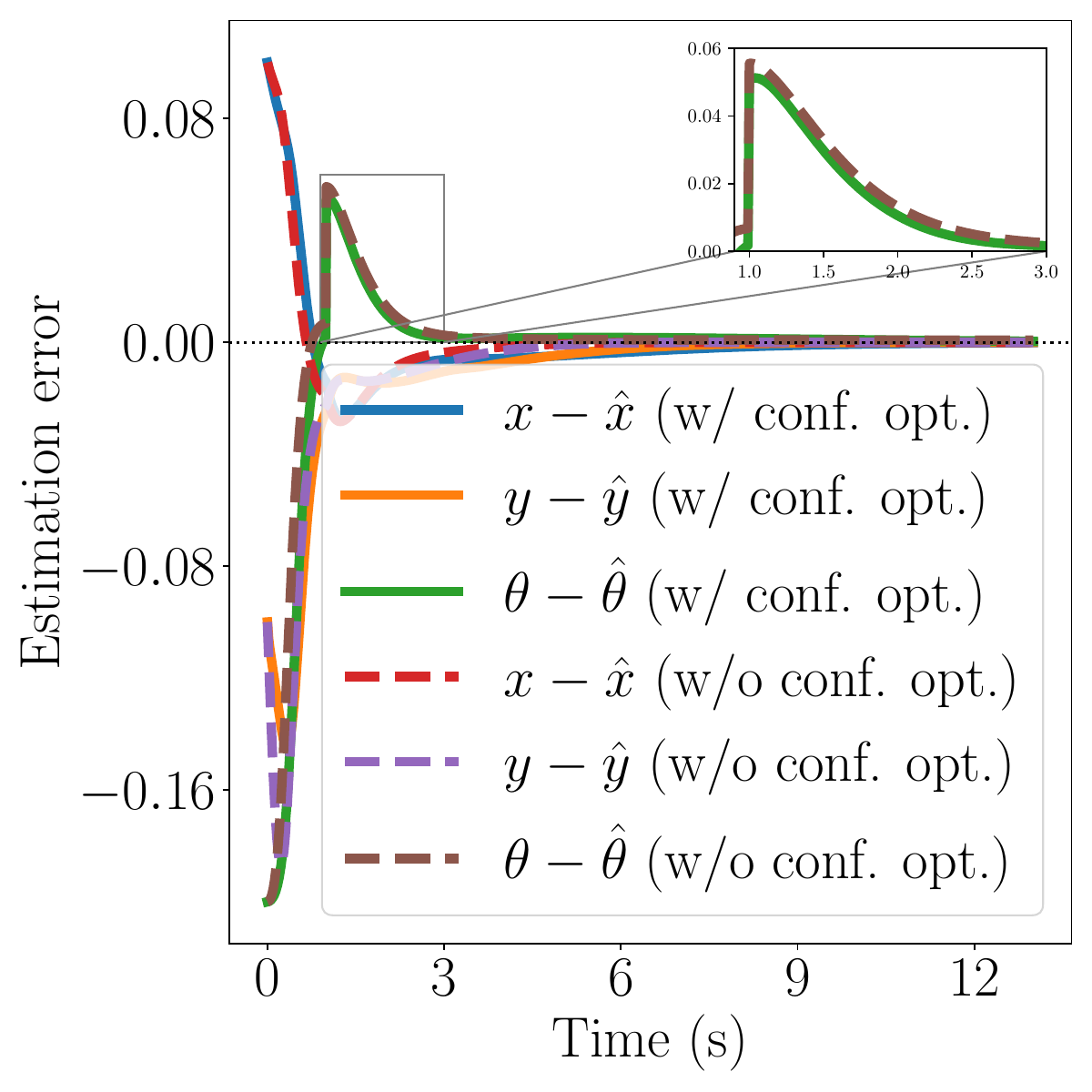}
        \caption{State estimation error.}
        \label{fig:p2_ekf_error}
    \end{subfigure}

    \caption{A unicycle control problem. \subref{fig:p2_trajs}: Trajectories with and without confidence optimization. \subref{fig:p2_control_wo_opt} and \subref{fig:p2_control_w_opt}: Controls without and with confidence optimization. \subref{fig:p2_eigenvalues} and \subref{fig:p2_ekf_error}: Comparison of the eigenvalues (of $P(t)$) and state estimation errors given by the observer.}
    \label{fig:p2}
\end{figure*}

\subsection{Unicycle}
\begin{figure}[t]
    \centering
    \includegraphics[width=0.18\textwidth]{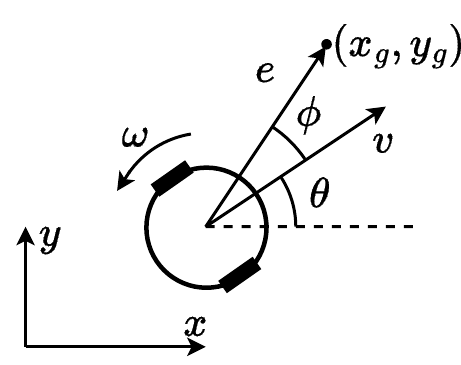}
    \caption{The unicycle.}
    \label{fig:eg2_unicycle}
\end{figure}
The dynamics of the unicycle are 
\begin{equation}\label{eq:eg2_dynamics}
     \dot{x} = v \cos(\theta), \quad
    \dot{y} = v \sin(\theta), \quad
    \dot{\theta} = \omega
\end{equation}
where $x$ and $y$ are coordinates of the unicycle in the world frame, and $\theta$ is the angle between the heading direction and the $x$-axis (see Fig.~\ref{fig:eg2_unicycle}). The outputs of the system are $x$ and $y$. Let the coordinates of the goal position be $(x_g, y_g)$. 
A feedback controller is proposed in \cite{aicardi1995closed}:
\begin{subequations}\label{eq:eg2_nominal_control}
\begin{align}
    v_n &= d_1 e \cos(\phi)\\
    \omega_n &= d_2 \phi + d_1 \cos(\phi) \sin(\phi) [\phi + d_3 (\phi + \theta)]/\phi
\end{align}
\end{subequations}
where $d_1, d_2, d_3 > 0$ are design parameters and 
\begin{subequations}
\begin{align}
    e &= \sqrt{(x-x_g)^2 + (y-y_g)^2},\\
    \phi &= \operatorname{atan2}(y_g-y, x_g-x) - \theta.
\end{align}
\end{subequations}

The task is to reach the goal position $(x_g, y_g) = (6,6)$ (by tracking the control given by $v_n$ and $\omega_n$) while avoiding a circular obstacle located at $(x_o, y_o) = (5.3, 4)$ with a radius of $r_o = 1.1$. The CBF for this task is defined as 
\begin{equation*}
    h(x,y) = (x-x_o)^2 + (y-y_o)^2 -r_o^2,
\end{equation*}
and we require the unicycle robot to stay in the region where $h(x,y) \geq 0$. 

For this example, we introduce an impulse disturbance at $t=1$~\si{s} to the dynamics of $\theta$. The height of this impulse is drawn from a uniform distribution $\mathbb{U}(-0.5,0.5)$. From Fig.~\ref{fig:p2_trajs}, we see that the control inputs without confidence optimization ($c_1 = 0$) cannot complete the navigation task, while the control inputs with confidence optimization ($c_1 = 1000$) meet the safety requirement and complete the task. From Figs.~\ref{fig:p2_control_wo_opt} and \ref{fig:p2_control_w_opt}, it may be observed that the linear velocity (in blue) is reduced in the beginning when $c_1 = 1000$ compared with the case where $c_1 = 0$. The executed angular speed $\omega$ (in orange) overlaps with the nominal angular speed $\omega_n$ (in red) in both Figs.~\ref{fig:p2_control_wo_opt} and \ref{fig:p2_control_w_opt} because the CBF for this task is independent of $\theta$. As we reduce the largest eigenvalue of $P(t)$ after $t =2$~\si{s} (see Fig.~\ref{fig:p2_eigenvalues}), a faster decrease in the estimation error $\theta - \hat{\theta}$ can be observed in Fig.~\ref{fig:p2_ekf_error}. It is also worth noticing that the maximum linear velocity that occurred in Fig.~\ref{fig:p2_control_w_opt} is smaller than that in Fig.~\ref{fig:p2_control_wo_opt}, which means that the control design with confidence optimization requires a smaller actuator and still can achieve safety and control objectives in this task.

\section{Conclusion}
This work addresses the synthesis of confidence-aware, safe, and stable control for control-affine systems in the output-feedback setting. We formulate two confidence-aware optimization problems, demonstrate their feasibility, and establish the Lipschitz continuity of the obtained solutions. Simulation studies indicate improvements in estimation accuracy and the fulfillment of safety and control requirements.

\bibliographystyle{IEEEtran}
\bibliography{master}

\end{document}